\documentclass[english,conference,final,a4paper]{IEEEtran}
 
\makeatletter

\usepackage[T1]{fontenc}
\usepackage[latin9]{inputenc}
\usepackage{amsthm}
\usepackage{amsmath}
\usepackage{amssymb}
\usepackage{graphicx}
\usepackage{dsfont}
\usepackage{cite}
\usepackage{amsmath}
\usepackage{array}
\usepackage{multirow}
\usepackage{caption}
\usepackage{color}

\usepackage{multicol}

\theoremstyle{plain}

\theoremstyle{plain}

\theoremstyle{plain}

\theoremstyle{plain}
\newtheorem{thm}{\protect\theoremname}
\theoremstyle{plain}
  
\theoremstyle{definition}

\theoremstyle{definition}

\theoremstyle{definition}

\makeatother
  
\usepackage{babel} 

\providecommand{\claimname}{Claim}
\providecommand{\lemmaname}{Lemma}
\providecommand{\propositionname}{Proposition}
\providecommand{\theoremname}{Theorem}
\providecommand{\corollaryname}{Corollary} 
\providecommand{\definitionname}{Definition}
\providecommand{\assumptionname}{Assumption}
\providecommand{\remarkname}{Remark}

\newcommand{\Shat}{\widehat{S}}

\newcommand{\kmax}{k_{\mathrm{max}}}

\newcommand{\ncheck}{\check{n}}

\newcommand{\ntil}{\widetilde{n}}

\newcommand{\Bernoulli}{\mathrm{Bernoulli}}

\newcommand{\pe}{P_{\mathrm{e}}}

\newcommand{\Yv}{\mathbf{Y}}

\newcommand{\Ac}{\mathcal{A}}
\newcommand{\Bc}{\mathcal{B}}
\newcommand{\Cc}{\mathcal{C}}

\newcommand{\EE}{\mathbb{E}}
\newcommand{\PP}{\mathbb{P}}

\usepackage{color}
\usepackage{enumitem}
\usepackage{float}

\floatstyle{ruled}
\newfloat{algorithm}{tbp}{loa}
\providecommand{\algorithmname}{Algorithm}
\floatname{algorithm}{\protect\algorithmname}

\makeatletter
\newcommand{\manuallabel}[2]{\def\@currentlabel{#2}\label{#1}}
\makeatother

\begin{document}

\title{An Efficient Algorithm for Capacity-Approaching \\ Noisy Adaptive Group Testing}

\author{
 \IEEEauthorblockN{Jonathan Scarlett}
  \IEEEauthorblockA{Deptartment of Computer Science \& Department of Mathematics \\
    National University of Singapore \\
    Email: scarlett@comp.nus.edu.sg }
} 

\maketitle

\begin{abstract}
    In this paper, we consider the group testing problem with adaptive test designs and noisy outcomes.  We propose a computationally efficient four-stage procedure with components including random binning, identification of bins containing defective items, $1$-sparse recovery via channel codes, and a ``clean-up'' step to correct any errors from the earlier stages.  We prove that the asymptotic required number of tests comes very close to the best known information-theoretic achievability bound (which is based on computationally intractable decoding), and approaches a capacity-based converse bound in the low-sparsity regime.
\end{abstract}

\long\def\symbolfootnote[#1]#2{\begingroup\def\thefootnote{\fnsymbol{footnote}}\footnote[#1]{#2}\endgroup}

\vspace*{-0.5ex}
\section{Introduction}

The group testing problem consists of determining a small subset $S$ of ``defective'' items within a larger set of items $\{1,\dotsc,p\}$, based on a number of possibly-noisy tests. This problem has a history in medical testing \cite{Dor43}, and has regained significant attention following new applications in areas such as communication protocols \cite{Ant11}, pattern matching \cite{Cli10}, and database systems \cite{Cor05}, and connections with compressive sensing \cite{Gil08,Gil07}.  Under a widely-adopted symmetric noise model, each test takes the form
\begin{equation}
    Y = \bigvee_{j \in S} X_j \oplus Z, \label{eq:gt_symm_model}
\end{equation}
where the test vector $X = (X_1,\dotsc,X_p) \in \{0,1\}^p$ indicates which items are included in the test, $Y$ is the resulting observation, $Z \sim \Bernoulli(\rho)$ for some $\rho \in \big(0,\frac{1}{2}\big)$, and $\oplus$ denotes modulo-2 addition.  The goal is to design a sequence of tests whose outcomes can be used to reliably recover $S$.

In the {\em adaptive} setting, a given test $X^{(i)}$ can be designed based on the previous outcomes $Y^{(1)},\dotsc,Y^{(i-1)}$.  While near-optimal adaptive designs have long been known in the noiseless setting \cite{Hwa72}, relatively less is known in the noisy setting, which is the focus of the present paper.  

In a recent work \cite{Sca18}, we developed both information-theoretic limits and performance bounds for practical algorithms, but with the gap between the two remaining significant.  In this paper, we substantially narrow this gap by providing a computationally efficient algorithm with a performance guarantee that nearly matches the best known information-theoretic achievability bound from \cite{Sca18}, which in turn nearly matches a converse bound in several regimes of interest.

\begin{figure*}
    \begin{centering}
        \includegraphics[width=0.88\columnwidth]{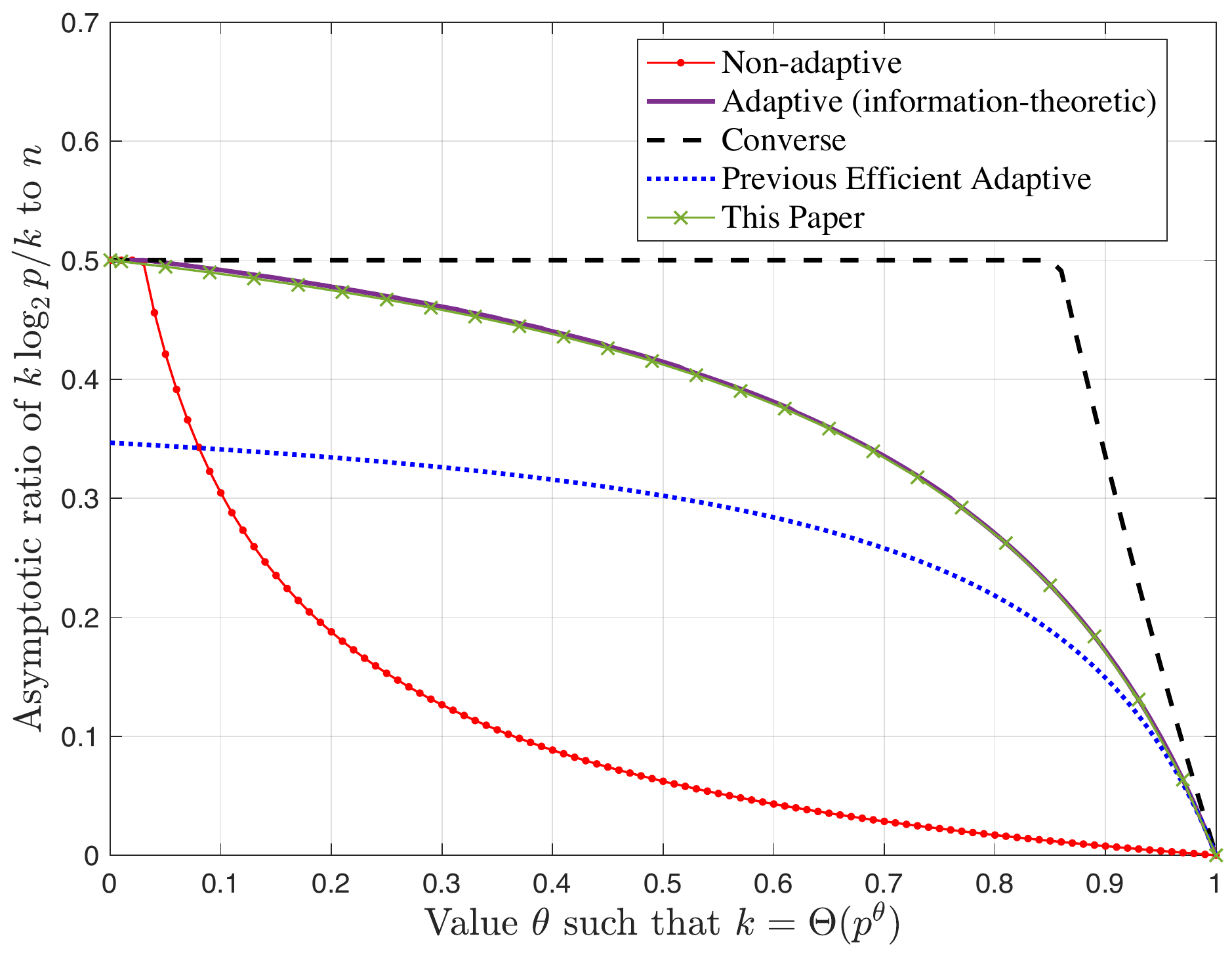} \qquad
        \includegraphics[width=0.88\columnwidth]{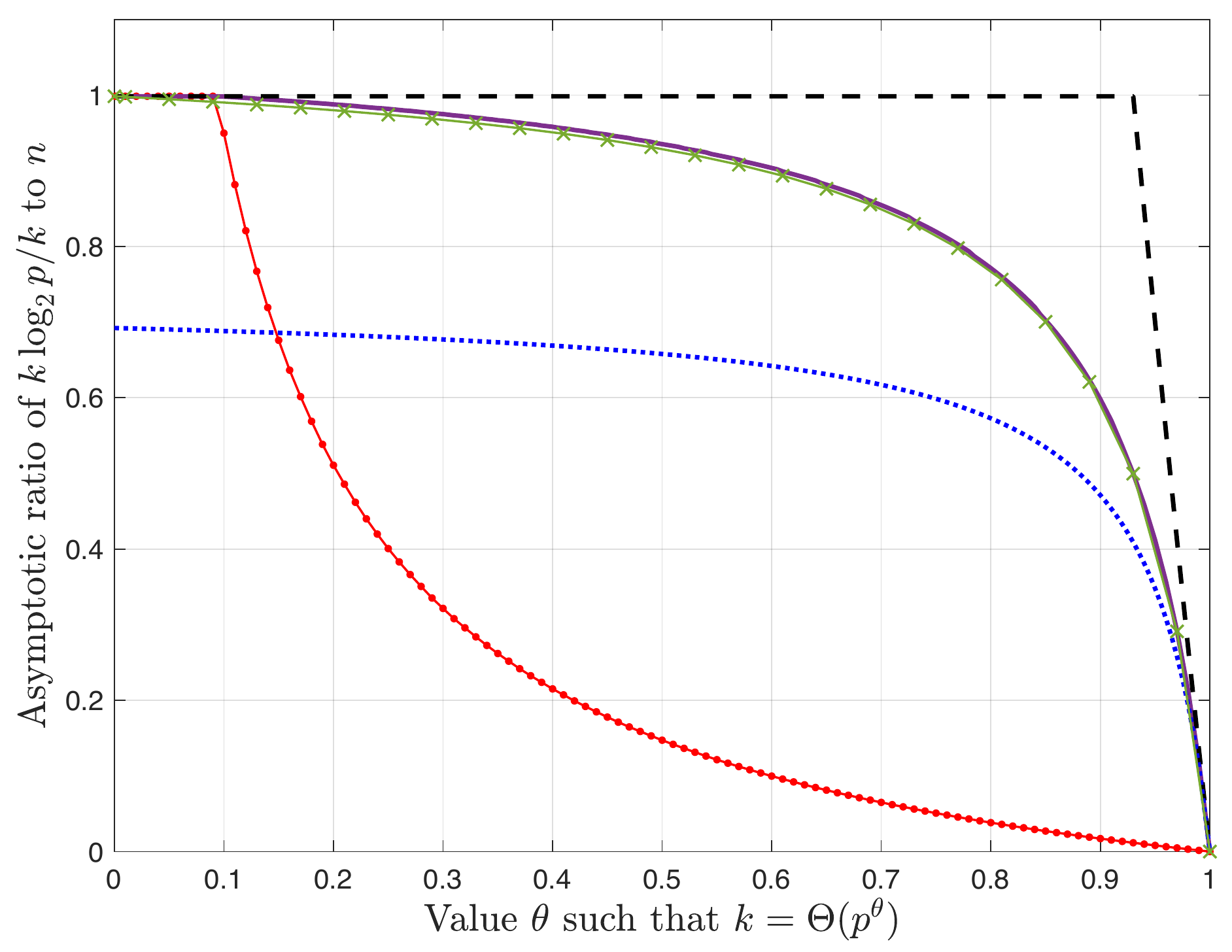}
        \par
    \end{centering}
    
    \caption{Asymptotic performance bounds for noisy group testing with noise level $\rho = 0.11$ (Left) and $\rho = 10^{-4}$ (Right). \label{fig:rho11}}
    \vspace*{-2ex}
\end{figure*}

\subsection{Problem Setup} \label{sec:setup}

We let the defective set $S$ be uniform on the ${p \choose k}$ subsets of $\{1,\dotsc,p\}$ of cardinality $k$.  An adaptive algorithm iteratively designs a sequence of tests $X^{(1)},\dotsc,X^{(n)}$, with $X^{(i)} \in \{0,1\}^p$.  The corresponding outcomes are denoted by $\Yv = (Y^{(1)},\dotsc,Y^{(n)})$ with $Y^{(i)} \in \{0,1\}$, and these follow the model \eqref{eq:gt_symm_model} with independence between samples.  Since we are in the adaptive setting, a given test is allowed to depend on all of the previous outcomes.

Given the tests and their outcomes, a \emph{decoder} forms an estimate $\Shat$ of $S$.  We consider the exact recovery criterion, in which the error probability is given by 
\begin{equation}
\pe := \PP[\Shat \ne S], \label{eq:pe}
\end{equation}
where the probability is with respect to the randomness of the defective set $S$, the tests $X^{(1)},\dotsc,X^{(n)}$ (if randomized), and the noisy outcomes $Y^{(1)},\dotsc,Y^{(n)}$.  We focus on the goal of vanishing error probability, i.e., $\pe \to 0$ as $n \to \infty$, in the sub-linear sparsity regime $k = \Theta(p^{\theta})$ with $\theta \in (0,1)$.

\subsection{Related work}

\noindent {\bf Information-theoretic limits.} The information-theoretic limits of group testing were first studied in the Russian literature \cite{Mal78,Mal80,Dya81}, and have recently become increasingly well-understood \cite{Ati12,Ald14,Sca15,Sca15b,Sca16b,Ald15}.  Among the existing works, the results most relevant to the present paper are as follows:
\begin{itemize}
    \item In the adaptive setting, it was shown by Baldassini {\em et al.} \cite{Bal13} that if the output $Y$ is produced by passing the noiseless outcome $U = \vee_{j \in S} X_j$ through a binary channel $P_{Y|U}$, then the number of tests for attaining $\pe \to 0$ must satisfy $n \ge \big(\frac{1}{C}k\log\frac{p}{k}\big)(1-o(1))$,\footnote{Here and subsequently, the function $\log(\cdot)$ has base $e$, and information measures have units of nats.} where $C$ is the Shannon capacity of $P_{Y|U}$.  For the symmetric noise model \eqref{eq:gt_symm_model}, this yields
    \begin{equation}
    n \ge \frac{k\log\frac{p}{k}}{\log 2 - H_2(\rho)} (1-o(1)), \label{eq:mi_conv}
    \end{equation}
    where $H_2(\rho) = \rho\log\frac{1}{\rho} + (1-\rho)\log\frac{1}{1-\rho}$.
    \item In the non-adaptive setting with symmetric noise, it was shown that an information-theoretic threshold decoder attains the bound \eqref{eq:mi_conv} for $k = o(p)$ under a partial recovery criterion \cite{Sca15b,Sca15}.  For exact recovery, a more complicated bound was also given in \cite{Sca15b} that matches \eqref{eq:mi_conv} when $k = \Theta(p^{\theta})$ for {\em sufficiently small} $\theta > 0$.
    \item In \cite{Sca18}, we provided information-theoretic achievability and converse bounds for the {\em noisy adaptive} setting that are often near-matching.  The achievability results are based on first achieving approximate recovery, and then using one or two extra stages of adaptivity to resolve the remaining errors in the estimate.
\end{itemize}

\noindent {\bf Non-adaptive algorithms.} Several non-adaptive noisy group testing algorithms have been shown to come with rigorous guarantees.  One of the building blocks of our algorithm is the {\em Noisy Combinatorial Orthogonal Matching Pursuit} (NCOMP) algorithm.  For each item, NCOMP checks the proportion of tests it was included in that returned positive, and declares the item to be defective if this number exceeds a suitably-chosen threshold.  This is known to provide optimal scaling laws for the regime $k = \Theta(p^{\theta})$ ($\theta \in (0,1)$) \cite{Cha11,Cha14}, albeit with somewhat suboptimal constants.  Improved constants have been provided via a technique known as {\em separate decoding of items} or {\em separate testing of inputs} \cite{Mal80,Sca17b}, as well as a noisy version of the Definite Defective (DD) algorithm \cite{Sca18b}.  The former is within a factor $\log 2$ of the information-theoretic limit as $\theta \to 0$, whereas the latter can provide better rates for higher values of $\theta$.

\noindent {\bf Adaptive algorithms.} As mentioned above, adaptive algorithms are relatively well-understood in the noiseless setting \cite{Hwa72,Dam12}.  To our knowledge, the first algorithm that was proved to achieve the optimal threshold $n = \big(k\log_2\frac{p}{k}\big)(1+o(1))$ for all $k = o(p)$ is Hwang's generalized binary splitting algorithm \cite{Hwa72,Du93}.  Various algorithms using limited rounds of adaptivity have also been proposed \cite{Mac98,Mez11,Dam12,Epp07}.

There are limited works on {\em noisy} adaptive algorithms.  In \cite{Cai13}, an algorithm called GROTESQUE was shown to provide optimal scaling laws in terms of samples and runtime, but no attempt was made to optimize the constant factors.  The above-mentioned work \cite{Sca18} also provided weaker bounds for a computationally efficient variant of the multi-stage algorithm.

\subsection{Contributions} \label{sec:contributions}

We provide a computationally efficient\footnote{See Section \ref{sec:runtime} for details on the decoding time.}  four-stage adaptive group testing algorithm which, as we will see in Section \ref{sec:result}, comes very close to matching the best known information-theoretic achievability bound from \cite{Sca18}.  Our algorithm makes use of various existing techniques (while combining and analyzing them in a novel manner), including the following:

    (i) We adopt the high-level approach of \cite{Sca18} of first achieving approximate recovery (i.e., finding a set $\hat{S}_1$ whose number of errors with respect to $S$ is an arbitrarily small fraction of $k$) and then performing two adaptive rounds to refine this initial estimate and achieve exact recovery.
    
    (ii) While \cite{Sca18} used non-adaptive methods to achieve approximate recovery, we use a two-stage adaptive method based on randomly binning the items, identifying which bins contain defectives, and then applying $1$-sparse recovery within those bins via a standard channel code.  To our knowledge, this idea was first proposed for a two-stage version of the GROTESQUE algorithm \cite{Cai13}.  However, \cite{Cai13} sought to achieve exact recovery from these two steps alone, resulting in a relatively large number of bins and requiring a channel code whose error probability decays exponentially in the code length.  Under the milder requirement of approximate recovery, we can use a much smaller number of bins, and adopt {\em any} capacity-achieving channel code.

\vspace*{-0.75ex}
\section{Main Result} \label{sec:result}
\vspace*{-0.25ex}

The main result of this paper is the following.  Here and subsequently, all asymptotic notation (e.g., $o(1)$) is as $p \to \infty$ (with $k \to \infty$ and $n \to \infty$ simultaneously).

\begin{thm} \label{thm:main}
    There exists a computationally efficient noisy four-stage group testing algorithm achieving exact recovery with $\pe \to 0$, using a number of tests satisfying the following when $k = \Theta(p^{\theta})$ for some $\theta \in (0,1)$:
        \begin{equation}
            n = \bigg( \frac{k \log \frac{p}{k}}{ \log 2 - H_2(\rho) } + \frac{k \log k}{D(\rho \| 1-\rho)}\bigg) (1+o(1)). \label{eq:n_final}
        \end{equation}
\end{thm}
\begin{proof}
    See Section \ref{sec:algo}.
\end{proof}

We note that this bound on the number of tests is in fact implicitly given in \cite{Sca18}, but only for a three-stage algorithm whose first stage uses a computationally intractable decoder (i.e., a brute force search over $p \choose k$ subsets) for non-adaptive designs.  Since $k \log k = \frac{\theta}{1-\theta}\big( k \log \frac{p}{k} \big)$ when $k = \Theta(p^{\theta})$, the second term in \eqref{eq:n_final} is negligible compared to the first in the limit $\theta \to 0$, so we asymptotically match the converse \eqref{eq:mi_conv}.

In Figure \ref{fig:rho11}, we compare Theorem \ref{thm:main} to other noisy adaptive group testing bounds from \cite{Sca18}, as well as a bound for the non-adaptive case from \cite{Sca15b}.  Among these, only the curve labeled ``Previous Efficient Adaptive'' was shown to be achievable with a computationally efficient algorithm, and this curve falls significantly short of that of Theorem \ref{thm:main}.  The curve labeled ``Adaptive (information-theoretic)'', corresponding to the above-mentioned three-stage algorithm, provides a very small (almost imperceptible) gain over that of Theorem \ref{thm:main}.  However, perhaps the most notable feature of this curve is that it achieves the capacity bound for {\em all sufficiently small $\theta$}, whereas Theorem \ref{thm:main} does so only in the limit as $\theta \to 0$ (despite coming very close at small $\theta$).

\section{Algorithm and its Analysis} \label{sec:algo}

As an important building block in our algorithm, we make use of the generic multi-stage procedure from \cite{Sca18}, described in Algorithm \ref{alg:steps_ref}.  This should be treated as an informal description, with the details given throughout the analysis.

\begin{algorithm}
    \caption*{ \manuallabel{pr:procedure}{1} \textbf{Algorithm 1:} $\mathbf{MultiStage}(\Ac_1,\Ac_2,p,k)$. \label{alg:steps_ref} }
    
    \smallskip
    \underline{Inputs}: Number of items $p$ and defectives $k$, adaptive or non-adaptive group testing algorithms $\Ac_1$ and $\Ac_2$.
    \smallskip
    
    \underline{Steps}:
    \begin{enumerate}
        \item[1.] Apply algorithm $\Ac_1$ to the ground set $\{1,\dotsc,p\}$ to find an estimate $\Shat_1$  of $S$ such that 
        \begin{equation}
            \max\{|\Shat_1 \backslash S|, |S \backslash \Shat_1|\} \le \alpha_1 k \label{eq:step1_req}
        \end{equation}
        with high probability, for some small $\alpha_1 > 0$.
        \item[2a.] Apply algorithm $\Ac_2$ to the reduced ground set $\{1,\dotsc,p\} \backslash \Shat_1$ to exactly identify the false negatives from the first step.  Let these items be denoted by $\Shat'_{2a}$.
        \item[2b.]  Test each item in $\Shat_1$ individually $\ncheck$ times (for suitably chosen $\ncheck$), and let $\Shat'_{2b} \subseteq \Shat_1$ contain the $k - \alpha_2 k$ items that returned positive the highest number of times, for some small $\alpha_2> 0$.
        \item[3.]  Test the items in $\Shat_1 \setminus \Shat'_{2b}$ (of which there are $\alpha_2 k$) individually $\ntil$ times (for suitably chosen $\ntil$), and let $\Shat'_3$ contain the items that returned positive at least $\frac{\ntil}{2}$ times.  The final estimate of the defective set is given by $\Shat := \Shat'_{2a} \cup \Shat'_{2b} \cup \Shat'_3$.
    \end{enumerate}
\end{algorithm}

\begin{algorithm}
    \caption*{ \manuallabel{pr:procedure}{2} \textbf{Algorithm 2:} $\mathbf{InnerAdaptive}(\Ac_3,\Cc,B,p)$. \label{alg:inner} }
    
    \smallskip
    \underline{Inputs}: Number of items $p$, number of bins $B$, channel code $\Cc$ containing $\frac{p}{B}$ codewords, adaptive or non-adaptive group testing algorithm $\Ac_3$. 
    \smallskip
    
    \underline{Steps}:
    
    \begin{enumerate}
        \item[1.] Partition the items $\{1,\dotsc,p\}$ into $B$ bins of size $\frac{p}{B}$ uniformly at random.
        \item[2.] Run the group testing algorithm $\Ac_3$ on the $B$ ``super-items'' formed in Step 1, where including super-item $j$ in a test is done by including every item in bin $j$.  
        \item[3.] For each bin indexed by $j \in \big\{ 1, \dotsc, B \big\}$ that returned positive, do the following:
        \begin{itemize}
            \item Apply non-adaptive group testing with $p' = \frac{p}{B}$ items, with a group testing matrix of size $n' \times p'$ constructed by arranging the $p'$ codewords of $\Cc$ in columns (where $n'$ is the code length).
            \item Use the test outcomes and a decoder for the channel code $\Cc$ to identify the single defective item in the bin.
        \end{itemize}
        \item[4.] Output $\hat{S}_{\rm inner}$ equaling the union of all single defective items identified in Step 3.
    \end{enumerate}
\end{algorithm}

In \cite{Sca18}, Algorithm \ref{alg:steps_ref} was used with both $\Ac_1$ and $\Ac_2$ being non-adaptive group testing algorithms, leading to a three-stage algorithm (steps 2a and 2b can be done in a single stage).  It was shown (see Theorem 3 and Footnote 2 in \cite{Sca18}) that if $\Ac_1$ achieves \eqref{eq:step1_req} using $n_1(\Ac_1)$ tests, then with $\Ac_2$ being a variant of the NCOMP algorithm \cite{Cha14}, the overall procedure succeeds with probability approaching one as long as
\begin{equation}
    n \ge \bigg( n_1(\Ac_1) + \frac{k \log k}{D(\rho \| 1-\rho)} \bigg)(1+\eta) \label{eq:n_init}
\end{equation}
for arbitrarily small $\eta > 0$, and suitably-chosen $\ncheck$ and $\ntil$ in the algorithm statement. In \cite{Sca18}, the ability to approach capacity at low sparsity levels was only shown when $\Ac_1$ is a computationally prohibitive exhaustive search decoder.  In contrast, we let $\Ac_1$ itself be an adaptive algorithm, which allows us to maintain computational efficiency.  This inner adaptive algorithm is shown in Algorithm \ref{alg:inner}; we proceed by describing these steps and providing their relevant analysis.

{\bf Analysis of first step.} In Step 1 of Algorithm \ref{alg:inner}, we partition the items into $B > 0$ equal-sized bins uniformly at random.\footnote{We ignore rounding issues, which have no impact on the result.}  Conditioned on a particular item being in a particular bin, we see that the probability of another particular item being in the same bin is at most $\frac{1}{B}$.  By the union bound, the probability of a particular defective item colliding with any of the other $k-1$ defectives is at most $\frac{k}{B}$, which behaves as $k^{-\epsilon} \to 0$ if $B = k^{1+\epsilon}$ for some $\epsilon > 0$.  Hence, the number of defectives $N_{\rm col}$ that are part of a collision satisfies $\EE[N_{\rm col}] \le k \cdot \big( \frac{k}{B} \big) = k^{1-\epsilon}$, and hence
\begin{equation}
    \PP\bigg[N_{\rm col} \ge \frac{\alpha_1}{3} k \bigg] \le \frac{3 k^{-\epsilon}}{\alpha_1} \to 0 \label{eq:first_step}
\end{equation}
for arbitrarily small $\alpha_1 > 0$ (note that $k \to \infty$ under the assumptions of Theorem \ref{thm:main}).  With this result in place, we seek to identify (most of) the bins containing at least one defective item, and then identify the defectives that did {\em not} collide using $1$-sparse recovery techniques. 

{\bf Analysis of second step.} To identify which bins contain one item, we apply a non-adaptive group testing procedure on $B$ ``super-items'', where including a super-item in a test amounts to including all of the items in the corresponding bin.  Since there are at most $k$ bins containing a defective, this amounts to recovering at most $k$ defectives among $B = k^{1+\epsilon}$ items.  In Appendix \ref{app:NCOMP_approx}, we explain how the analysis of the variant of NCOMP used in \cite{Sca18} can be adapted to ensure approximate recovery with $O\big( k\log\frac{B}{k} \big)$ tests (as opposed to $O(k \log B)$ for exact recovery).  Specifically, using $O\big( k\log\frac{B}{k} \big)$ tests with a sufficiently large implied constant, we can construct a set $\Bc$ of bin indices such that
\begin{equation}
    \max\{|\Bc \backslash \Bc^*|, |\Bc^* \backslash \Bc|\} \le \frac{\alpha_1}{3} k \label{eq:bin_guarantee}
\end{equation}
with probability approaching one, where $\Bc^*$ is the set of true defective bins. Here $\alpha_1 > 0$ is the same as in \eqref{eq:first_step}.

Note that since $B = k^{1 + \epsilon}$, the number of tests $O\big( k\log\frac{B}{k} \big)$ used in this step simplifies to
\begin{equation}
    n_{{\rm inner},2} = O( \epsilon k \log k ). \label{eq:n_inner2}
\end{equation}
As we will see, the contribution to the overall number of tests can be made negligible by taking $\epsilon$ to be small.


{\bf Analysis of the third step.} In Step 3 of Algorithm \ref{alg:inner}, for every bin that returned positive, we apply $1$-sparse group testing using non-adaptive methods.  If the bin truly does contain a single defective item, then identifying it is a relatively easy problem, because in the $1$-sparse setting the observation vector is simply equal to the column of the test matrix indexed by the unique defective item, plus the noise term.  Therefore, by letting the columns of the non-adaptive testing matrix be the codewords of a channel code, we precisely recover a channel coding problem.  We could use capacity-achieving expander codes as in \cite{Cai13}, but in contrast to \cite{Cai13} it in fact suffices for our purposes to use {\em any} capacity-achieving code.

Since there are $p' = \frac{p}{B}$ columns (which ensures $p' \to \infty$ under the choice $B = k^{1+\epsilon}$ and scaling $k = \Theta(p^{\theta})$, as long as $\epsilon$ is small enough), and the capacity is $\log 2 - H_2(\rho)$, we deduce that any given $1$-sparse sub-routine succeeds with probability approaching one provided that
\begin{equation}
    n' \ge (1+\eta) \frac{\log \frac{p}{B}}{ \log 2 - H_2(\rho) } \label{eq:n'}
\end{equation}
for arbitrarily small $\eta > 0$, and also provided that the bin truly does contain a single defective item.

We proceed by adopting a pessimistic analysis: If we run $1$-sparse recovery on a bin with no defective items, we assume that this results in some item erroneously being marked as defective.  In addition, if we run $1$-sparse recovery on a bin with multiple defectives, we assume that this results in all of these defectives being marked as non-defective, {\em and} an additional non-defective erroneously being marked as defective.

In light of this pessimistic view, we analyze the number of {\em additional} mistakes caused upon performing $1$-sparse recovery on the bins that do contain a single defective.  There are at most $k$ such bins, and since we have not specified the convergence rate of the error probability, a union bound may make the overall error probability bound large.  Instead, letting $\delta_p \to 0$ denote the individual error probability, we note that the average number of errors is at most $\delta_p k$.  Hence, by Markov's inequality, the probability of the number of errors exceeding $\frac{\alpha_1}{3} k$ is at most $\frac{3 \delta_p}{\alpha_1}$, which is vanishing since $\delta_p \to 0$.

{\bf Wrapping up: Total number of mistakes.} Combining the above, the contributions of errors in the final estimate of the defective set produced by Algorithm \ref{alg:inner} behaves as follows with probability approaching one:
\begin{itemize}
    \item In Step 1, the collisions cause at most $\frac{\alpha_1}{3}$ false negatives for the collided items, and at most $\frac{\alpha_1}{6}$ false positives from applying $1$-sparse recovery to a bin with multiple defectives ({\em cf.}, \eqref{eq:first_step}).  Note that we replace $3$ by $6$ in the latter expression because each bin with collisions contains two or more defectives.
    \item In Step 2, the missed bins $\Bc^* \setminus \Bc$ cause at most $\frac{\alpha_1}{3}$ false negatives, and the false positive bins $\Bc \setminus \Bc^*$ cause at most  $\frac{\alpha_1}{3}$ false positives ({\em cf.}, \eqref{eq:bin_guarantee}).
    \item In Step 3, we have at most $\frac{\alpha_1 k}{3}$ decoding errors, meaning the number of errors of each type is at most $\frac{\alpha_1 k}{3}$.
\end{itemize}
Combining these bounds, we find that the required condition \eqref{eq:step1_req} in Algorithm \ref{alg:steps_ref} holds with probability approaching one.

{\bf Wrapping up: Total number of tests.} We count the number of tests used above to obtain an expression for $n_1(\Ac_1)$ in \eqref{eq:n_init}.  Adding $n_{{\rm inner},2}$ in \eqref{eq:n_inner2} with $k$ times the expression for $n'$ in \eqref{eq:n'}, substituting $B = k^{1 + \epsilon}$, and noting that $\epsilon$ may be arbitrarily small, we readily deduce the overall number of tests stated in Theorem \ref{thm:main}.

    Our algorithm has four rounds of adaptivity overall, since Algorithm \ref{alg:steps_ref} is a ``three-stage'' algorithm, but its first stage is replaced by Algorithm \ref{alg:inner}, which performs tests in two stages.

\subsection{Discussion: Decoding Time and Limitations} \label{sec:runtime}

The runtimes of the various sub-routines used in the decoding procedure are outlined as follows:
\begin{itemize}
    \item Steps 2b and 3 of Algorithm \ref{alg:steps_ref} use a trivial decoding procedure, so their total runtime matches the corresponding number of tests, $O(k \log k)$ \cite{Sca18}.
    \item In Step 2 of Algorithm \ref{alg:inner}, running NCOMP on $k^{1+\epsilon}$ items with $O(\epsilon k \log k)$ tests incurs a runtime of $O( \epsilon k^{2+\epsilon} \log k )$.
    \item Letting $\tau_{\Cc}$ be the decoding time of the code $\Cc$, Step 3 of Algorithm \ref{alg:inner} incurs decoding time $O(k^{1+\epsilon} \tau_{\Cc})$.  For example, if $\Cc$ can be decoded in time linear in the block length, then $\tau_{\Cc} = O(\log \frac{p}{k})$.
    \item For $\Ac_2$ in Step 2a of Algorithm \ref{alg:steps_ref}, using NCOMP as proposed in \cite{Sca18} incurs decoding time $O( \alpha_1 k p \log p )$, which may be much higher than the previous dot points since $k \ll p$.  However, by replacing NCOMP by the recently-proposed bit-mixing coding (BMC) method \cite[Sec.~IV]{Bon19a}, this can be reduced to $O( (\alpha_1 k)^2 \cdot \log(\alpha_1 k) \cdot \log p )$, so that the overall decoding time is polynomial in $k \log p$ and does not incur any linear dependence on $p$.
\end{itemize}


A limitation of our analysis is that the convergence of $\pe$ to zero may be very slow.  For instance, our error probability bound contains terms of the form $\frac{k^{-\epsilon}}{\alpha_1}$, where we take {\em both} $\epsilon$ and $\alpha_1$ to be arbitrarily small at the end of the proof.  Comparing noisy adaptive group testing strategies in {\em finite-size} systems remains an interesting direction for future work.

\appendices

\section{Approximate Recovery for NCOMP} \label{app:NCOMP_approx}

A slight variant of NCOMP \cite{Cha11} was proposed in \cite{Sca18} for achieving $n = O(k \log p)$ scaling for recovering $k$ defectives out of $p$ items when $k$ is only known up to a constant factor, i.e., $k \in [c_0 \kmax,\kmax]$ for some $\kmax = \Theta(p^{\theta})$.\footnote{In addition, as noted in \cite{Sca18}, if only an upper bound $k \le \kmax'$ is known, one can move to the regime that it is known to within a factor of two by adding an extra $\kmax'$ ``dummy'' defective items.  This is merely a trick to simplify the theoretical analysis, and is unlikely to be useful in practice.}  Here we outline how this extension of \cite{Sca18} can be modified to achieve $n = O\big(k \log \frac{p}{k}\big)$ scaling for {\em approximate recovery} with $k = o(p)$.  Specifically, we seek to produce an estimate $\hat{S}$ such that $\max\{|\Shat \backslash S|, |S \backslash \Shat|\} \le \alpha k$ with probability approaching one, for arbitrarily small $\alpha > 0$.

Under Bernoulli testing, in which each item is placed in each test independently with probability $\frac{\nu}{\kmax}$ for some $\nu > 0$, the analysis in \cite[Appendix C]{Sca18} proceeds as follows:
\begin{itemize}
    \item Letting $N'_j$ denote the number of tests in which item $j$ is included, show that
    \begin{equation}
        \PP\bigg[ N'_j \le \frac{n\nu}{2\kmax} \bigg] \le e^{-\Theta(1) \frac{n}{\kmax}}. \label{eq:ncomp1}
    \end{equation}
    \item Letting $N'_{j,1}$ be the number of the $N'_j$ tests including $j$ that returned positive, show that conditioned on $N'_j = n'_j$ for some $n'_j > \frac{n\nu}{\kmax}$, it holds for any defective $j$ that
    \begin{equation}
        \PP\big[ N'_{j,1} < (1 - \rho - \Delta)n'_j \big] \le e^{-\Theta(1) \frac{n\nu}{2\kmax}}, \label{eq:ncomp2}
    \end{equation}
    and for any non-defective $j$ that
    \begin{equation}
        \PP\big[ N'_{j,1} \ge (1 - \rho - \Delta)n'_j \big] \le e^{-\Theta(1) \frac{n\nu}{2\kmax}} \label{eq:ncomp3}
    \end{equation}
    when $\Delta > 0$ is sufficiently small (but constant).
\end{itemize}
In light of these bounds,  NCOMP simply declares the items $j$ satisfying $N'_{j,1} \ge (1 - \rho - \Delta)N'_j$ to be defective.  A simple union bound over $p$ items in \eqref{eq:ncomp1}, $k$ items in \eqref{eq:ncomp2}, and $p-k$ items in \eqref{eq:ncomp3} yields vanishing error probability when $n = \Theta(k \log p)$ with a sufficiently large implied constant.  

The modification for approximate recovery is straightforward: The average number of mistakes (false positives or false negatives) is upper bounded by the sum of $p$ times \eqref{eq:ncomp1}, $k$ times \eqref{eq:ncomp2}, and $p-k$ times \eqref{eq:ncomp3}.  In particular, if $k = o(p)$ and $n = \Theta(k \log \frac{p}{k})$ with a sufficiently large implied constant, then the average number of mistakes behaves as $o(k)$; then, by Markov's inequality, the probability of the number of mistakes exceeding $\alpha k$ tends to zero for any fixed $\alpha > 0$.

\paragraph*{Acknowledgment}
This work was supported by an NUS Early Career Research Award. 

\vspace*{-0.5ex}
\bibliographystyle{IEEEtran}
\bibliography{../JS_References}

\begin{thebibliography}{10}
\providecommand{\url}[1]{#1}
\csname url@samestyle\endcsname
\providecommand{\newblock}{\relax}
\providecommand{\bibinfo}[2]{#2}
\providecommand{\BIBentrySTDinterwordspacing}{\spaceskip=0pt\relax}
\providecommand{\BIBentryALTinterwordstretchfactor}{4}
\providecommand{\BIBentryALTinterwordspacing}{\spaceskip=\fontdimen2\font plus
\BIBentryALTinterwordstretchfactor\fontdimen3\font minus
  \fontdimen4\font\relax}
\providecommand{\BIBforeignlanguage}[2]{{%
\expandafter\ifx\csname l@#1\endcsname\relax
\typeout{** WARNING: IEEEtran.bst: No hyphenation pattern has been}%
\typeout{** loaded for the language `#1'. Using the pattern for}%
\typeout{** the default language instead.}%
\else
\language=\csname l@#1\endcsname
\fi
#2}}
\providecommand{\BIBdecl}{\relax}
\BIBdecl

\bibitem{Dor43}
R.~Dorfman, ``The detection of defective members of large populations,''
  \emph{Ann. Math. Stats.}, vol.~14, no.~4, pp. 436--440, 1943.

\bibitem{Ant11}
A.~Fern\'andez~Anta, M.~A. Mosteiro, and J.~Ram\'on Mu\~{n}oz, ``Unbounded
  contention resolution in multiple-access channels,'' in \emph{Distributed
  Computing}.\hskip 1em plus 0.5em minus 0.4em\relax Springer Berlin
  Heidelberg, 2011, vol. 6950, pp. 225--236.

\bibitem{Cli10}
R.~Clifford, K.~Efremenko, E.~Porat, and A.~Rothschild, ``Pattern matching with
  don't cares and few errors,'' \emph{J. Comp. Sys. Sci.}, vol.~76, no.~2, pp.
  115--124, 2010.

\bibitem{Cor05}
G.~Cormode and S.~Muthukrishnan, ``What's hot and what's not: Tracking most
  frequent items dynamically,'' \emph{ACM Trans. Database Sys.}, vol.~30,
  no.~1, pp. 249--278, March 2005.

\bibitem{Gil08}
A.~Gilbert, M.~Iwen, and M.~Strauss, ``Group testing and sparse signal
  recovery,'' in \emph{Asilomar Conf. Sig., Sys. and Comp.}, Oct. 2008, pp.
  1059--1063.

\bibitem{Gil07}
A.~C. Gilbert, M.~J. Strauss, J.~A. Tropp, and R.~Vershynin, ``One sketch for
  all: Fast algorithms for compressed sensing,'' in \emph{Proc. ACM-SIAM Symp.
  Disc. Alg. (SODA)}, New York, 2007, pp. 237--246.

\bibitem{Hwa72}
F.~Hwang, ``A method for detecting all defective members in a population by
  group testing,'' \emph{J. Amer. Stats. Assoc.}, vol.~67, no. 339, pp.
  605--608, 1972.

\bibitem{Sca18}
J.~Scarlett, ``Noisy adaptive group testing: Bounds and algorithms,'' 2019,
  accepted to {\em IEEE Trans. Inf. Theory}.

\bibitem{Mal78}
M.~Malyutov, ``\BIBforeignlanguage{English}{The separating property of random
  matrices},'' \emph{\BIBforeignlanguage{English}{Math. Notes Acad. Sci.
  {USSR}}}, vol.~23, no.~1, pp. 84--91, 1978.

\bibitem{Mal80}
M.~B. Malyutov and P.~S. Mateev, ``Screening designs for non-symmetric response
  function,'' \emph{Mat. Zametki}, vol.~29, pp. 109--127, 1980.

\bibitem{Dya81}
A.~G. D'yachkov, ``Error probability bounds for the symmetrical model of the
  design of screening experiments,'' \emph{Prob. Inf. Transm.}, 1982.

\bibitem{Ati12}
G.~Atia and V.~Saligrama, ``Boolean compressed sensing and noisy group
  testing,'' \emph{IEEE Trans. Inf. Theory}, vol.~58, no.~3, pp. 1880--1901,
  March 2012.

\bibitem{Ald14}
M.~Aldridge, L.~Baldassini, and K.~Gunderson, ``Almost separable matrices,''
  \emph{J. Comb. Opt.}, pp. 1--22, 2015.

\bibitem{Sca15}
J.~Scarlett and V.~Cevher, ``Limits on support recovery with probabilistic
  models: An information-theoretic framework,'' \emph{IEEE Trans. Inf. Theory},
  vol.~63, no.~1, pp. 593--620, 2017.

\bibitem{Sca15b}
------, ``Phase transitions in group testing,'' in \emph{Proc. ACM-SIAM Symp.
  Disc. Alg. (SODA)}, 2016.

\bibitem{Sca16b}
------, ``Converse bounds for noisy group testing with arbitrary measurement
  matrices,'' in \emph{IEEE Int. Symp. Inf. Theory}, Barcelona, 2016.

\bibitem{Ald15}
M.~Aldridge, ``The capacity of {B}ernoulli nonadaptive group testing,''
  \emph{IEEE Trans. Inf. Theory}, vol.~63, no.~11, pp. 7142--7148, 2017.

\bibitem{Bal13}
L.~Baldassini, O.~Johnson, and M.~Aldridge, ``The capacity of adaptive group
  testing,'' in \emph{IEEE Int. Symp. Inf. Theory}, July 2013, pp. 2676--2680.

\bibitem{Cha11}
C.~L. Chan, P.~H. Che, S.~Jaggi, and V.~Saligrama, ``Non-adaptive probabilistic
  group testing with noisy measurements: Near-optimal bounds with efficient
  algorithms,'' in \emph{Allerton Conf. Comm., Ctrl., Comp.}, Sep. 2011, pp.
  1832--1839.

\bibitem{Cha14}
C.~L. Chan, S.~Jaggi, V.~Saligrama, and S.~Agnihotri, ``Non-adaptive group
  testing: Explicit bounds and novel algorithms,'' \emph{IEEE Trans. Inf.
  Theory}, vol.~60, no.~5, pp. 3019--3035, May 2014.

\bibitem{Sca17b}
J.~Scarlett and V.~Cevher, ``Near-optimal noisy group testing via separate
  decoding of items,'' \emph{IEEE Trans. Sel. Topics Sig. Proc.}, vol.~2,
  no.~4, pp. 625--638, 2018.

\bibitem{Sca18b}
J.~Scarlett and O.~Johnson, ``Noisy non-adaptive group testing: A
  {(near-)}definite defectives approach,'' 2018,
  https://arxiv.org/abs/1808.09143.

\bibitem{Dam12}
P.~Damaschke and A.~S. Muhammad, ``Randomized group testing both query-optimal
  and minimal adaptive,'' in \emph{Int. Conf. Current Trends in Theory and
  Practice of Computer Science}.\hskip 1em plus 0.5em minus 0.4em\relax
  Springer, 2012, pp. 214--225.

\bibitem{Du93}
D.~Du and F.~K. Hwang, \emph{Combinatorial group testing and its
  applications}.\hskip 1em plus 0.5em minus 0.4em\relax World Scientific, 2000,
  vol.~12.

\bibitem{Mac98}
A.~J. Macula, ``Probabilistic nonadaptive and two-stage group testing with
  relatively small pools and {DNA} library screening,'' \emph{J. Comb. Opt.},
  vol.~2, no.~4, pp. 385--397, 1998.

\bibitem{Mez11}
M.~M{\'e}zard and C.~Toninelli, ``Group testing with random pools: Optimal
  two-stage algorithms,'' \emph{IEEE Trans. Inf. Theory}, vol.~57, no.~3, pp.
  1736--1745, 2011.

\bibitem{Epp07}
D.~Eppstein, M.~T. Goodrich, and D.~S. Hirschberg, ``Improved combinatorial
  group testing algorithms for real-world problem sizes,'' \emph{SIAM Journal
  on Computing}, vol.~36, no.~5, pp. 1360--1375, 2007.

\bibitem{Cai13}
S.~Cai, M.~Jahangoshahi, M.~Bakshi, and S.~Jaggi, ``Efficient algorithms for
  noisy group testing,'' \emph{IEEE Trans. Inf. Theory}, vol.~63, no.~4, pp.
  2113--2136, 2017.

\bibitem{Bon19a}
S.~Bondorf, B.~Chen, J.~Scarlett, H.~Yu, and Y.~Zhao, ``Sublinear-time
  non-adaptive group testing with {$O(k \log n)$} tests via bit-mixing
  coding,'' 2019, https://arxiv.org/abs/1904.10102.

\end{thebibliography}

\end{document}